\documentclass[year=26]{fmcad}

\usepackage{graphicx}
\usepackage{hyperref}
\hypersetup{colorlinks = true,linkcolor = black}

\usepackage{color}

\usepackage{cite}
\usepackage{tcolorbox}
\usepackage{scalerel,stackengine}
\usepackage{multicol}

\usepackage{amsmath}

\usepackage{amsthm}

\newcommand{\lp}{\left(}
\newcommand{\rp}{\right)}

\newcommand{\until}[1]{\,U_{#1}\,}
\newcommand{\release}[1]{\,R_{#1}\,}
\newcommand{\eventually}[1]{F_{#1}\,}
\newcommand{\always}[1]{G_{#1}\,}

\newcommand{\rob}[1]{[\![ #1 ]\!]}
\newcommand{\robapprox}[1]{\{\!\!\{ #1 \}\!\!\}\!}
\newcommand{\maxminsem}{max/min}
\newcommand{\linearsem}{semi-linear}

\usepackage{hhline} 
\usepackage{multirow}
\usepackage[table]{xcolor}

\usepackage{subcaption} 

\usepackage[algoruled,linesnumberedhidden,commentsnumbered,vlined]{algorithm2e}
\DontPrintSemicolon
\SetKwProg{Proc}{Procedure}{}{}
\SetKw{KwTo}{:}
\usepackage{algorithmic}

\usepackage[shortlabels]{enumitem} 

\newcommand{\bounds}{BOUNDS}

\newcommand{\variantI}{V1}
\newcommand{\variantIV}{V1-UNB}
\newcommand{\variantVI}{V1-LIN}
\newcommand{\variantVII}{V1-LIN-UNB}
\newcommand{\variantUNB}{\variantIV}
\newcommand{\variantLIN}{\variantVI}
\newcommand{\variantLINUNB}{\variantVII}

\newcommand{\signal}{s}
\newcommand{\sigvals}{\stVec}
\newcommand{\sigtimes}{\timeVec}
\newcommand{\siglen}{n}
\newcommand{\formulalen}{m}

\usepackage{amsfonts}
\newcommand{\reals}{\Re}

\newcommand{\nnreals}{\reals_{\geq 0}}

\newcommand{\nats}{\Ne}

\newcommand{\defeq}{:=}

\newcommand{\timeVec}{\mathbf{t}}

\newcommand{\stPt}{x}
\newcommand{\stVec}{\mathbf{\stPt}}
\newcommand \formula \phi 

\newcommand{\Ne}{\mathbb{N}}

\renewcommand{\Re}{\mathbb{R}}
\theoremstyle{plain} 
\newtheorem{theorem}{Theorem}[section]

\theoremstyle{definition}
\newtheorem{definition}{Definition}[section]

\usepackage{dblfloatfix}

\usepackage{url}

\begin{document}
\title{Privacy-Preserving Robust Monitoring of Dense-Time Temporal Logic}

\author{\IEEEauthorblockN{Charles Koll\orcid{0000-0001-5941-250X}}
\IEEEauthorblockA{\textit{Oregon State University}\\
Corvallis, OR, USA \\
kollch@oregonstate.edu}
\and
\IEEEauthorblockN{Mike Rosulek\orcid{0009-0003-5206-6191}}
\IEEEauthorblockA{\textit{Oregon State University}\\
Corvallis, OR, USA \\
rosulekm@oregonstate.edu}
\and
\IEEEauthorblockN{Houssam Abbas\orcid{0000-0002-8096-2618}}
\IEEEauthorblockA{\textit{Oregon State University}\\
Corvallis, OR, USA \\
houssam.abbas@oregonstate.edu}
}

\maketitle

\begin{abstract}
Motivated by the use of third-party testing, this paper presents a privacy-preserving implementation of the robustness monitoring of Signal Temporal Logic (STL).
STL is a specification language for systems whose output traces have real-valued timestamps -- most prominently, cyber-physical systems. Its robust semantics, which assign real rather than binary truth values, are an essential ingredient of systems testing and runtime monitoring.
In a number of system testing scenarios, the designer and the tester are separate parties. This testing-as-a-service model is common in software, avionics, and semiconductors, and is growing in importance in the red-teaming of machine learning models.
In such a model, the designer seeks to preserve the privacy (or secrecy) of its design, while the tester seeks to preserve the privacy of its collateral, including the detailed specification suite.
This paper presents a secure multi-party implementation of the robust semantics using arithmetic circuits.
This implementation allows one party, which provides a secret trace, and a second party which provides a secret formula, to jointly compute the robustness without either party being able to learn the other party's secret input.
We present a number of algorithmic modifications to the robustness computation to enable an efficient multi-party robustness computation. We also present new, semi-linear, robust semantics for STL that take advantage of the efficiencies of arithmetic circuits. 
The experimental results demonstrate that not only is this privacy-preserving monitor efficient on its own, but also that it can be used in a testing loop without affecting the speed at which bugs are found during robustness-guided system testing.
\end{abstract}

\IEEEpeerreviewmaketitle

\section{Introduction}
Monitoring, also known as runtime verification, is a key routine in the testing and runtime operation of Cyber-Physical Systems (CPS). 
The basic setup is one where an executing system generates an output signal $\signal$, which is a function of time, that is supposed to satisfy some specification $\formula$ in an appropriate temporal logic. 
A \emph{monitor} receives that signal and determines whether it satisfies the specification.
Monitoring is used at runtime (during system operation) to detect spec violations, 
or at design time, during system testing, to guide the test case generator. 
When monitoring or testing systems with physical components (e.g., cyber-physical systems), the specification is typically formalized in a dense-time logic, i.e., one which places real-time constraints on the signals and is interpreted over traces with real-valued timestamps. Moreover, one is especially interested in \textit{robust} satisfaction, i.e. in obtaining a real-valued measure of how well the trace satisfies the specification or how badly it violates it, not just in a binary satisfies/violates answer.
This paper uses Signal Temporal Logic, or STL, as a specification logic \cite{maler2004monitoring}, and the max/min robust semantics of STL \cite{fainekos2009robustness} for quantitative satisfaction. These are very commonly used options in cyber-physical systems monitoring and testing. 

In some applications, the monitored trace and the formal specification are owned by two different parties, and there's a need to keep each of them a secret known only to its owner. 
For example, \emph{Testing-as-a-Service} (TaaS), which is the practice of outsourcing some testing activities to an outside party, is common in industries like semi-conductor \cite{cadencesite,movellussite,rambussite}, aerospace \cite{collinstest}, fuel cell batteries \cite{avl}, and software~\cite{markresearch,rainforest,testlio,councill99swtesting}.
There is significant economic incentive in building novel TaaS models: the global TaaS market in software was estimated at USD 4.54 billion in 2023~\cite{gvr24} and at USD 2.5 billions in semi-conductors~\cite{lucintel24market}.
The design company has clear incentives for keeping its design secret. In particular, this implies keeping the system's output traces secret, since otherwise the tester might learn to map input traces (which it generates) to output traces. 
The testing company, on the other hand, possesses refined and targeted correctness specifications which it developed over years of experience with different designs. These specifications steer the test case generator towards effective tests. The testing company is thus incentivized to keep them secret. 
A key challenge of TaaS is therefore the protection of the intellectual properties (IP) of both the design and the testing companies, namely the design's output signals (and the design that produced them), and the tester's correctness specifications. 
\begin{figure}[t]
\centering
  \includegraphics[width=\linewidth]{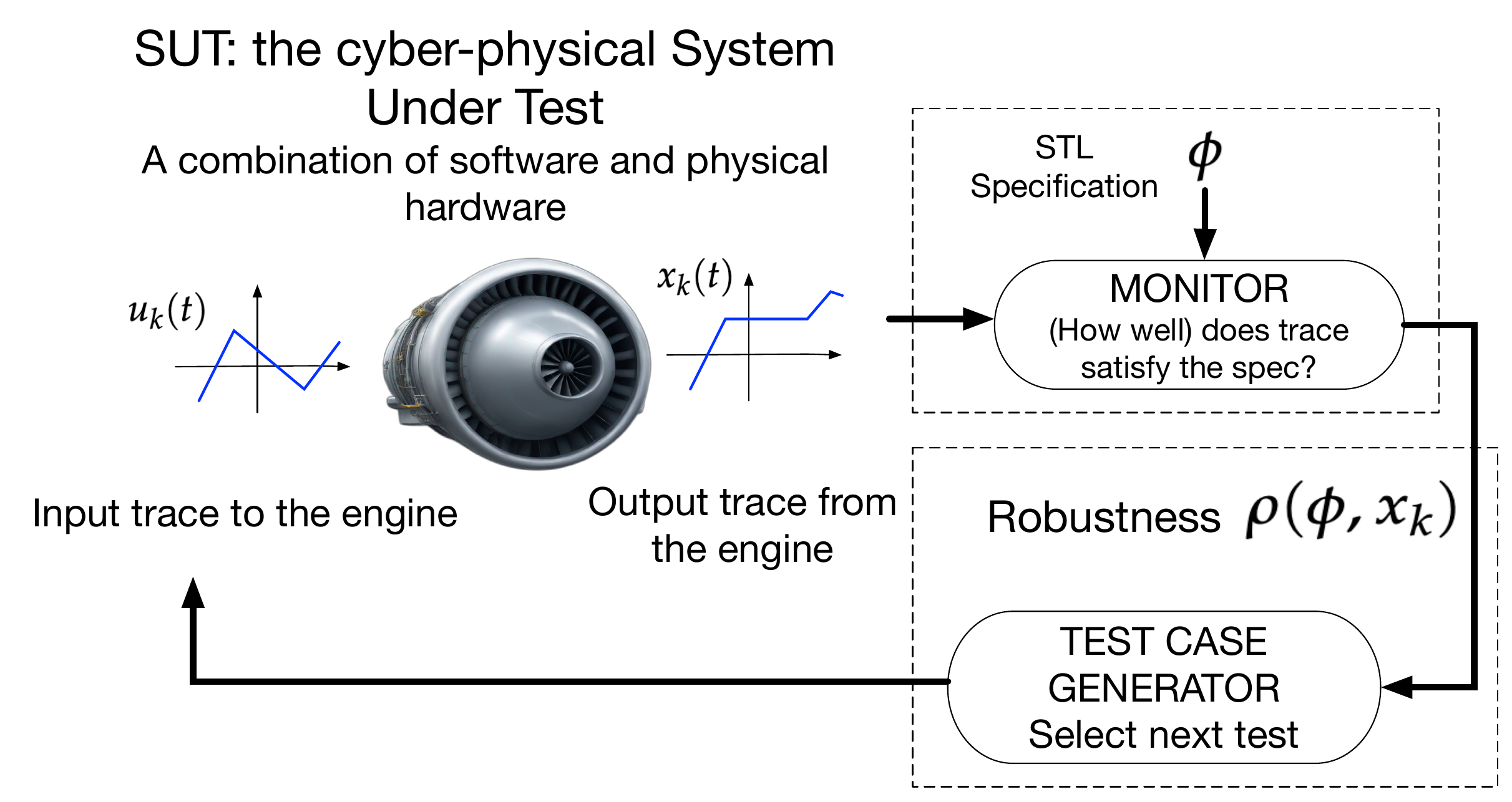}
  \caption{CPS monitoring-and-testing loop. A CPS System Under Test (SUT) takes an input signal and produces an output signal. The robustness of the output signal w.r.t. to an STL specification is computed by the monitor and used to guide the test generator. The signal is a proprietary secret to the system designer, while the formula is a proprietary secret of the testing company.}
  \label{fig:loop}
\end{figure}

Thus we are faced with the following problem: given two parties, Designer who owns a secret dense-time trace and Verifier who owns an STL formula, how can they compute the robustness of the trace relative to the formula, so that both learn the robustness, but Verifier does not learn anything else about the monitored trace, and Designer does not learn anything about what formula the trace is monitored against?
Such an implementation is called \textit{privacy-preserving}, or simply \textit{private}.
Non-Disclosure Agreements (NDAs) provide penalties for leaking trace and formula beyond the two parties, but cannot ensure that design and formula secrets are not leaked between parties.

In this paper, we use a cryptographic technique called secure multi-party computation \cite{evans2018pragmatic}, aka MPC, to create a privacy-preserving implementation of STL robustness monitoring. Our work is \emph{offline}, meaning that the robustness monitoring only occurs once Designer has a full trace. We seek to answer the following research questions:
\begin{enumerate}
    \item Can we develop an MPC implementation of robustness monitoring that is fast enough for use in a robustness-guided testing loop?
    \item Does private monitoring negatively affect the rate at which robustness-guided testing finds bugs? By bug finding rate, we mean the number of tests needed to find a bug. This is distinct from the time it takes to monitor a given output trace, produced by the system in response to a given test.
\end{enumerate}
MPC is known to add significant computational overhead, so merely dropping an algorithm into a general-purpose MPC framework will not be enough. Algorithmic and domain-specific theoretical modifications will be necessary, and these constitute the core of this paper.

\noindent
\emph{\textbf{Related Work.}\:}
There are very few works on private monitoring. 
The concept of private offline monitoring was introduced for LTL in \cite{abbas19prv}, which also sketched a number of approaches to preserving the privacy of the monitored signal -- but not of the LTL formula.
Banno et al. \cite{banno22prv} define a scheme for preserving the privacy of both a \emph{safety-only} LTL formula and a signal, using fully homomorphic encryption \cite{chillotti2020tfhe} and the LTL formula's representation as a DFA (Deterministic Finite Automaton).
We work with the full Signal Temporal Logic (STL) required for CPS designs, so both our signals and formulas are significantly different. 
In particular, STL formulas in dense time do not have a DFA representation, so work on oblivious DFA evaluation, like \cite{blanton10obliviousdfa} and \cite{sasakawa2014oblivious}, is not applicable here.
Along similar lines, works such as \cite{waga2025oblivious} rely on characteristics of discrete-time STL for their approach, which do not work in our dense-time setting.
Henzinger et al. \cite{henzinger2025privacy} have the same goal of private runtime verification, but approach the problem differently -- they encode the formula specification as a circuit, while we implement a functionality that takes the formula as an input. Additionally, their experiments focus on simple functionality such as comparators and adders, whereas we compute robustness and assume more complex specifications.


\noindent
\emph{\textbf{Contributions.}\:}
This work makes the following contributions:
\begin{enumerate}
    \item An approach for an MPC-based approach to private STL robustness monitoring for any formula up to a given depth;
    \item A novel robust semantics of STL that computes more efficiently than traditional robust semantics in arithmetic circuits;
    \item An experimental confirmation of the practicality of private monitoring using the MP-SPDZ framework \cite{mp-spdz}. 
    \item An experimental confirmation that the private computation does not degrade the rate at which robustness-driven testing finds bugs.
\end{enumerate}

Along with this document we have also submitted an artifact containing code for our experiments.

\noindent
\emph{\textbf{Organization.}\:}
\autoref{sec:prelims} provides background regarding STL, robustness, and multi-party computation; \autoref{sec:optimizing robustness} describes our variants of robustness monitoring algorithms which are efficient for MPC, including new robustness semantics; \autoref{sec:experiments} presents experimental results; and \autoref{sec:conclusion} concludes.

\section{Technical Background}\label{sec:prelims}

\subsection{Signal Temporal Logic (STL) and its Robustness}
Signal Temporal Logic (STL) \cite{maler2004monitoring} is commonly used in the formal verification of cyber-physical systems. STL is interpreted over dense-time signals.
\begin{definition}
\label{def:signal}
Fix constants $C \in \nnreals$, $T \in \nnreals$ and $\siglen \in \Ne$.
    A \emph{signal} $\signal$ is a tuple $(\sigvals, \sigtimes) \in [-C, C]^{\siglen} \times [0, T]^{\siglen}$. 
    This says that at timestamp $\sigtimes(i)$ the signal value is $\sigvals(i)$. The timestamps obey $0 = \sigtimes(1) < \sigtimes(i) < \sigtimes(i + 1) < \sigtimes(\siglen) = T$ for all $1 < i < \siglen - 1$. We write $t_i$ and $x_i$ as shorthand for $\sigtimes(i)$ and $\sigvals(i)$, respectively.
\end{definition}
The constants $C$, $T$, and $\siglen$ are the same for all signals. We also define an inverse function $\sigtimes^{-1}(I) \defeq \{i \in \nats \mid \sigtimes(i) \in I\}$, where $I \subseteq [0, T]$. This function returns the indices of sampling times in $I$.

The STL syntax (in negation normal form) is as follows:
\begin{equation*}
    \begin{split}
        \phi \defeq& \top \mid \bot \mid x \geq c \mid \lnot(x \geq c) \mid \phi_1 \lor \phi_2 \mid \phi_1 \land \phi_2 \\
        &\mid \phi_1 \until{I} \phi_2 \mid \phi_1 \release{I} \phi_2
    \end{split}
\end{equation*}
where $c \in [-C, C]$ is a constant, $I \subseteq [0, T]$ is a closed nonnegative interval, ``$x \geq c$'' is an atomic proposition, and $\phi_1$ and $\phi_2$ are STL formulas.
The Eventually and Always operators are defined via Until and Release respectively -- Eventually: $\eventually{I} \phi \defeq \top \until{I} \phi$, and Always: $\always{I} \phi \defeq \bot \release{I} \phi$.

There exist several notions of robustness \cite{fainekos2009robustness} -- we focus on the \emph{\maxminsem} semantics, which are the most commonly used in testing and monitoring. 

\begin{definition}[Robustness \cite{fainekos2009robustness}]
\label{def:maxmin semantics}
    Let $\phi$ be an STL formula, $\signal$ a signal, and $i$ a sampling index of $\signal$. Write $K(i) \defeq \sigtimes^{-1}(\sigtimes(i) + I)$ for brevity, where $I$ is an interval $I = [a, b]$ and $\sigtimes(i) + I$ is defined as the interval $[\sigtimes(i) + a, \sigtimes(i) + b]$. The \maxminsem{} robustness $\rob{\phi}(\signal, i)$ is computed as follows.
    \begin{align*}
        \rob{\top}(\signal, i) &\defeq +\infty, \qquad \quad \rob{\bot}(\signal, i) \defeq -\infty         \\
        \rob{x \geq c}(\signal, i) &\defeq \sigvals(i) - c,\quad \rob{\lnot(x \geq c)}(\signal, i) \defeq c - \sigvals(i) \\
        \rob{\phi_1 \lor \phi_2}(\signal, i) &\defeq \max(\rob{\phi_1}(\signal, i), \rob{\phi_2}(\signal, i)) \\
        \rob{\phi_1 \land \phi_2}(\signal, i) &\defeq \min(\rob{\phi_1}(\signal, i), \rob{\phi_2}(\signal, i)) \\
        \rob{\phi_1 \until{I} \phi_2}(\signal, i) &\defeq \max_{j \in K(i)}\ \min\lp\rob{\phi_2}(\signal, j), \min_{i \leq k < j}\rob{\phi_1}(\signal, k)\rp \\
        \rob{\phi_1 \release{I} \phi_2}(\signal, i) &\defeq \min_{j \in K(i)} \max\lp\rob{\phi_2}(\signal, j), \max_{i \leq k < j}\rob{\phi_1}(\signal, k)\rp
    \end{align*}
\end{definition}

From these semantics it follows that 
$\rob{\eventually{I} \phi}(\signal, i) = \displaystyle{\max_{j \in K(i)}}\rob{\phi_2}(\signal, j)$ and $\rob{\always{I} \phi}(\signal, i) = \displaystyle{\min_{j \in K(i)}}\rob{\phi_2}(\signal, j)$.

A positive robustness value indicates satisfaction, a negative value indicates violation, 0 is inconclusive, and the magnitude $|\rob{\phi}(\signal, i)|$ indicates degree of satisfaction or violation \cite{fainekos2009robustness}. 

\subsection{Dynamic Program for Robustness Calculation}
\label{sec:dp-taliro}
Robustness calculation is done using a dynamic program (DP) known as DP-TALIRO\cite{yang2013dynamic}. We begin describing it with an example.

\begin{figure}[t]
    \centering
\begin{subfigure}[b]{0.24\textwidth}
    \caption{\mbox{Signal} $\signal$.}
    \centering
    \begin{tabular}{|c||c|c|c|c|}
        \hline
        $\sigtimes$  & 0.0 & 5.0 & 7.0 & 10.0 \\
        \hline
        $\sigvals$  & 3.0 & 11.0 & -2.0 & -3.0 \\
        \hline
    \end{tabular}    
    \label{tab:dp example signal}
\end{subfigure}
\hfill
\begin{subfigure}[t]{0.4\textwidth}
\centering
    \begin{tikzpicture}[
  every node/.style = {inner sep=3pt, anchor=center},
  level distance=1cm,
  level 1/.style={sibling distance=2cm}
]

  \node {$\until{[4,9]}$}
    child { node {$x \geq 0$} }
    child { node {$\neg(x \geq 10)$} };

\end{tikzpicture}
    \caption{The syntax tree of formula $\formula$ given in \eqref{eq:dp example formula}.}
    \label{fig:dp example formula}
\end{subfigure}
\hfill
\begin{subfigure}[b]{0.4\textwidth}
    \caption{DP table $R$ for computing $\rob{\formula}(\signal, \cdot)$.}
    \centering
    \begin{tabular}{|l||c|c|c|}
        \hline
        & $\until{[4,9]}$ & $x \geq 0$ & $\neg(x \geq 10)$ \\
        \hhline{|=#=|=|=|}
        $t_1 = 0.0$  & $3$ & $3$ & $7$ \\
        \hline
        $t_2 = 5.0$  & $-2$ & $11$ & $-1$ \\
        \hline
        $t_3 = 7.0$  & $-\infty$ & $-2$ & $12$ \\
        \hline
        $t_4 = 10.0$ & $-\infty$ & $-3$ & $13$ \\
        \hline
    \end{tabular}    
    \label{tab:dp example table}
\end{subfigure}
\caption{Robustness computation for the example in \autoref{sec:dp-taliro}.}
\end{figure}

\emph{\textbf{Monitoring Example.}\:}
Consider the signal shown in \autoref{tab:dp example signal} and the following formula:
\begin{equation}\label{eq:dp example formula}
    \phi: (x \geq 0)~\until{[4,9]}\, \neg(x \geq 10)
\end{equation}

The DP constructs a table $R$ with rows labeled by the timestamps $t_i$ and columns labeled by the sub-formulas $\phi_j$ of $\phi$, as shown in \autoref{tab:dp example table}. 
The DP then fills in the cells, starting with the bottom right cell, working right to left and bottom to top. 
Cell $R[t_i,\phi_j]$ contains the value $\rob{\phi_j}(\signal,i)$.
The calculation ends with $\rob{\phi}(\signal, 0)$ filled into the top left cell of the table. 
For convenience, cell $R[t_i, \phi_j]$ will be referred to simply as $R[i,j]$. We now detail this calculation.

Cell $R[4,3]$ contains $\rob{\neg (x \geq 10)}(\signal,10) = 10 - (-3)=13$. 
The next cell to the left $R[4,2]$ similarly contains $\rob{\neg (x \geq 0)}(\signal,10) = -3 -0=-3$. 
The next cell, $R[4,1]$, is the first that requires previously computed cells. To compute $\rob{\phi}(\signal, 4)$, first the samples that fall within the interval $t_4 + [4,9]$ are identified using the \bounds{} routine 
(see Appendix \ref{appendix:bounds} for details). 
In this example, no samples exist in the interval $10 + [4,9]$, so $R[4,1]$ is set to $-\infty$ as per the semantics.

The DP moves up one row to $t_3$ and performs similar operations. 
The first two cells, containing atoms, are filled with $R[3,3]=10 - (-2.0) = 12$ and $R[3,2]=-2.0 - 0 = -2$. 
The leftmost cell, $R[3,1]$, requires a \bounds{} calculation which again yields no samples in the interval $7+[4,9]$ and so the result is $-\infty$.
When the DP reaches the second row, $t_2$, on the other hand, \bounds{} determines that sample $t_4=10$ falls within $5+[4,9]$. The semantics then dictate that the robustness is (after simplification) $R[2,1]=\min(R[2,2], R[3,2], R[4,3]) = -2$.
The algorithm continues in this way until it fills $R[1,1]$.
\qed

\begin{algorithm}
    \KwIn{$\phi$ in breadth-first traversal of its syntax tree $[\phi_1, \dots, \phi_m]$, signal $s=[(x_1, t_1), \dots, (x_n, t_n)]$}
    \KwOut{Robustness $\rob{\phi}(\signal,1)$}
    \BlankLine

        \tcc{$b_l[j],b_u[j]$ hold lower \& upper indices of column $j$'s interval}
        $b_l, b_u := [\infty, \dots, \infty], [-\infty, \dots, -\infty]$ (length $m$ each)\;
        \tcc{Iterate row by row through all cells of the table $R$}
        \ShowLn\label{algoln:outer loop and diagonal}\For{$i := n \KwTo 1$, $j := m \KwTo 1$}{
            \Switch{$\phi_j$}{
                \ShowLn\label{algoln:atoms}\lCase{$\top~/~\bot$}{$R[i, j] := +\infty/~-\infty$}                
                \ShowLn\label{algoln:atoms2}\lCase{$p(c)~/~\neg p(c)$}{$R[i, j] := x_i - c~/~c - x_i$}
                \ShowLn\label{algoln:nonatoms}\lCase{$\phi_{k_1} \lor \phi_{k_2}$}{$R[i, j] := \max(R[i, k_1], R[i, k_2])$}
                \ShowLn\label{algoln:nonatoms2}\lCase{$\phi_{k_1} \land \phi_{k_2}$}{$R[i, j] := \min(R[i, k_1], R[i, k_2])$}
                \uCase{$\phi_{k_1} \until{I} \phi_{k_2}$}{
                    \tcc{Some special cases can be handled quickly}
                    \lIf{$i = n \text{ and } 0 \in I$}{$R[i, j] := R[i, k_2]$}
                    \lElseIf{$i = n \text{ and } 0 \notin I$}{$R[i, j] := -\infty$}
                    \lElseIf{$I = [0, T]$}{$R[i, j] := \max(R[i, k_2], \min(R[i, k_1], R[i+1, j]))$}
                    \Else(\tcp*[h]{General case}){
                        \ShowLn\label{algoln:bounds1} $b_l[j], b_u[j] := \textbf{\bounds}(I, i, [t_1, \dots, t_n], n, b_l[j], b_u[j])$\;
                        $tmp_{min} := \min_{i \leq i' < b_l[j]} R[i', k_1]$\;
                        $R[i, j] := -\infty$\;
                        \ShowLn\label{algoln:secondinnerloop}\For{$i' := b_l[j] \KwTo b_u[j]$}{
                            $R[i, j] := \max(R[i, j], \min(R[i', k_2], tmp_{min}))$\;
                            $tmp_{min} := \min(tmp_{min}, R[i', k_1])$\;
                        }
                    }
                }
                \uCase{$\phi_{k_1} \release{I} \phi_{k_2}$}{
                    \lIf{$i = n \text{ and } 0 \in I$}{$R[i, j] := R[i, k_2]$}
                    \lElseIf{$i = n \text{ and } 0 \notin I$}{$R[i, j] := +\infty$}
                    \lElseIf{$I = [0, T]$}{$R[i, j] := \min(R[i, k_2], \max(R[i, k_1], R[i+1, j]))$}
                    \Else{
                        \ShowLn\label{algoln:bounds2} $b_l[j], b_u[j] := \mathbf{\bounds}(I, i, [t_1, \dots, t_n], n, b_l[j], b_u[j])$\;
                        $tmp_{max} := \max_{i \leq i' < b_l[j]} R[i', k_1]$\;
                        $R[i, j] := +\infty$\;
                        \For{$i' := b_l[j] \KwTo b_u[j]$}{
                            $R[i, j] := \min(R[i, j], \max(R[i', k_2], tmp_{max}))$\;
                            $tmp_{max} := \max(tmp_{max}, R[i', k_1])$\;
                        }
                    }
                }
                \ShowLn\label{algoln:unnecessary1} \uCase{$\eventually{I} \phi_{k}$ or $\always{I} \phi_{k}$}{
                \tcp{omitted for brevity}
                }
            }
        }
        \Return{$R[1, 1]$}
    \caption{DP-TALIRO \cite{yang2013dynamic}}
    \label{algo:dp-taliro}
\end{algorithm}

The general DP-TALIRO algorithm is now provided, since our MPC implementation will modify parts of it to obtain an efficient calculation. 
DP-TALIRO contains two components: the primary robustness calculation, shown in \autoref{algo:dp-taliro}, and a subroutine \bounds{} that selects which elements of a signal fall within the bounds of a temporal interval; this is shown in Appendix~\ref{appendix:bounds}. The two components perform the row-by-row, right-to-left table filling described in the example. Correctness of DP-TALIRO is proved in \cite{yang2013dynamic}.

The inputs to DP-TALIRO are a sequence $[(x_1, t_1), \dots, (x_{\siglen}, t_{\siglen})]$ representing the signal $\signal$, provided by the Designer party; 
and an STL formula $\phi$ represented as a sequence of its sub-formulas, obtained in a breadth-first traversal of the formula's syntax tree. This is provided by the Tester party.
Let that sequence be $[\phi_1, \dots, \phi_{\formulalen}]$, with $\phi_1 = \phi$ being the root. 
We assume that the formula's syntax tree is complete -- this is necessary for the MPC implementation, which requires all inputs to have a fixed size. There is no loss of generality in this assumption since a formula without a complete tree can be padded with $\top$ to form a complete tree without affecting the robust semantics of \autoref{def:maxmin semantics}.

The subroutine \bounds{} (called at lines \ref{algoln:bounds1} and \ref{algoln:bounds2} of \autoref{algo:dp-taliro}) takes two arguments, an index $i$ and an interval $I \subseteq [0, T]$, and returns the minimum and maximum values in $\sigtimes^{-1}(\sigtimes(i)+I)$. It does this by simply iterating through the signal starting at $i$. Once found, these two indices are passed to the primary robustness calculation to be used for an Until. 

Note that even though Eventually and Always operators are derived from the Until and Release, it is more efficient to provide a custom implementation for them, as shown on line \ref{algoln:unnecessary1}. 
As will be seen in Section \ref{sec:optimizing robustness}, this is no longer true when implementing the algorithm as a multi-party computation. 

Let $\siglen$ be the length of the input signal, let $\formulalen$ be the number of nodes of the input formula's syntax tree, and let $c$ be the maximum number of sampling points for any interval $I$ in the formula, at most $\siglen$. The complexity of the DP-TALIRO algorithm is $O(\siglen \cdot \formulalen \cdot c)$ \cite{yang2013dynamic}.

\subsection{Secure Multi-Party Computation (MPC)}
\label{sec:mpc background}
Secure Multi-Party Computation (MPC) is a branch of cryptography that tackles the following problem:
Consider two parties, Alice and Bob; Alice holds private value $x$, and Bob holds private value $y$.
Together they wish to compute the value $f(x,y)$, for some agreed upon functionality or algorithm $f$ that takes $x$ and $y$ as inputs. 
For example, $x$ can be a dense-time signal, $y$ can be an STL formula, and $f$ can compute the robustness of $x$ with respect to $y$.
\textit{But each would like the other to learn nothing more than the output of $f$} -- in particular, that they learn nothing about each other's inputs, even if Alice or Bob behaves adversarially.
A \emph{secure} or \emph{privacy-preserving} interactive MPC protocol is a sequence of steps that Alice and Bob take in order to 
\begin{itemize}
    \item compute and learn the value $f(x,y)$, [Correctness]
    \item and learn nothing else. [Privacy]
\end{itemize}

What does it mean formally that the parties `learn nothing more than the output of $f$'? For instance, under what assumptions on the parties' computational power?
These questions have all received a rigorous treatment in the MPC literature -- e.g. \cite{evans2018pragmatic}, \cite{Cramer_Damgrd_Nielsen_2015}, and \cite{canetti01UC}.
A self-contained overview of security definitions is far beyond the scope of the current paper, but we do give a minimal set of definitions in Appendix \ref{sec:security guarantees} for completeness. The informal statement of security is as follows:
\begin{definition}[Secure protocol (Intuitive)]
\label{def:intuitive security}
    An MPC protocol implements functionality $f$ in a secure manner w.r.t. semi-honest parties (parties that implement the protocol without deviation, but might try to learn something extra based on the messages they see) if it is highly improbable that the parties could learn anything from executing the protocol, that they couldn't learn anyway just from knowing their inputs and the value $f(x,y)$.
\end{definition}
The intuitive understanding of security is sufficient to follow the rest of this paper, but whenever we reference a `secure protocol', formally this means secure in the sense of the definitions of Appendix \ref{sec:security guarantees}.
The following is a well-known result. 
Arithmetic circuits are explained below.
\begin{theorem}
\label{thm:general feasibility}
    Any functionality $f$ that can be computed in time that is polynomial in the input sizes can also be computed by any secure interactive MPC protocol. In particular, it can be computed by a secure interactive MPC protocol for arithmetic circuits \cite{Cramer_Damgrd_Nielsen_2015}. 
\end{theorem}

There are indeed general-purpose \emph{compilers} that accept a description of $f$ and implement such an interactive protocol.
In this work, we use one such compiler called MP-SPDZ \cite{mp-spdz}, in which $f$ is specified in a subset of Python.
While compilers like MP-SPDZ give a turnkey solution for MPC, the resulting protocol is often prohibitively inefficient in practice.
It is usually necessary to tailor the computation $f$ to be more ``MPC-friendly.''
The following are the most important factors affecting the performance of MPC protocols.
They will guide our modifications to DP-TALIRO to make it more efficient \textit{within an MPC}.

\textbf{Circuits:} all MPC protocols relevant for our purposes first express the desired functionality $f$ as a strictly combinational \emph{circuit} and then, using cryptographic techniques, ``blindly evaluate'' the circuit on private data.
The circuit may be boolean (consisting of \textsc{and, or, not} gates) or \textbf{arithmetic} (addition and multiplication gates in a finite field): both can implement any computable function \cite{Cramer_Damgrd_Nielsen_2015}.
MPC protocols represent intermediate secret values in a way that is only compatible with one of these two classes of operations.
Because robustness calculation uses max, min and addition, which are most naturally expressed as numerical calculations, \emph{we opt for an arithmetic circuit protocol in this paper.}

\textbf{Circuits imply fixed input size}: given that $f$ is represented as a circuit, this circuit has a fixed number of input wires, so the inputs $x,y$ must have fixed sizes. In this paper this means all signals have the same length, and all formulas have the same depth and have complete syntax trees. As explained earlier, a formula with an incomplete tree can be padded without changing its semantics.

\textbf{Flow control / branching:}
Circuits have no flow control, so a corollary of the above is that \emph{flow control based on private data is extremely expensive for MPC protocols}.
For example, consider the expression ``[if $b$ then $C_0(x)$ else $C_1(x)$]''.
If $b$ is a secret value (one that cannot be inferred from the overall function's output), then the parties must not learn which branch of this if-statement is active.
To convert the expression to a circuit, one must include sub-circuits for both $C_0$ and $C_1$.
Both sub-circuits must be evaluated, and the ``correct'' result is chosen by a multiplexer controlled by bit $b$.

Any form of flow control based on private data leads to a significant blow-up in the size of the resulting circuit.
In fact, MP-SPDZ's dialect of Python explicitly forbids static branching on private data, having a for-loop with data-dependent bounds, etc.
Typically one must unroll all branching (as in the if-statement example) and pad for-loops to their worst-case maximum length, etc.

\textbf{Linearity is free:}
The parties in an MPC protocol will collectively hold a ``blind'' representation of the values along every wire of the circuit.
These representations are \emph{linearly homomorphic},
meaning that they can perform linear operations on private data without any communication between the parties.
In an arithmetic circuit, the linear operation is the field's addition operation.
Consequently, multiplication by a \emph{public} constant is also free.
The freeness of addition will motivate the creation of a new `semi-linear' semantics for STL.

\textbf{Circuit depth and parallelism:}
MPC uses interactive protocols, where the  parties communicate. 
In an arithmetic MPC protocol, the parties must exchange protocol messages for each (non-free) gate of the circuit.
Furthermore, if the output of gate $g$ influences the input of gate $g'$, then the communication for $g$ and $g'$ must happen \emph{sequentially.}
Gates that do not influence each other in the circuit can have their messages sent in parallel.
So an entire layer of a circuit can be processed in the same network round-trip, but consecutive layers require additional round-trips.
Thus, often the limiting factor in MPC is the \emph{depth} of the circuit (the maximum path-length from an input gate to an output gate).

It is preferable, therefore, to tailor the computation $f$ so that its corresponding circuit $f$ is ``wide and shallow,'' with minimal depth.
\textit{When we describe a computation as ``parallel'', we are referring to how amenable it is to reducing circuit depth,
and the resulting effect on the MPC protocol's network latency.}
This kind of parallelism is orthogonal to whether the MPC participants perform their (local) cryptographic operations in a single thread or multiple threads.

\section{A Secure Multi-Party Computation of Robustness}
\label{sec:optimizing robustness}

Robustness is computed by a dynamic program, which takes time polynomial in the size of the inputs $\signal$ and $\formula$. 
Thus by Theorem \ref{thm:general feasibility} in Section \ref{sec:mpc background}, any interactive MPC protocol to securely evaluate $\rob{\formula}$ using an arithmetic circuit can be used. 
This simple observation is captured in the following for later reference.
\begin{theorem}
\label{thm:robustness can be mpced}
    The robustness functionality described in Def.~\ref{def:maxmin semantics}, with inputs $\formula$ and $\signal$, can be computed in a secure manner by an interactive MPC protocol for arithmetic circuits.
\end{theorem}

\begin{algorithm}
    \SetNoFillComment
    \LinesNumberedHidden
    \KwIn{$\phi$ in breadth-first traversal of its syntax tree $[\phi_1, \dots, \phi_m]$, signal $s = [(x_1, t_1), \dots, (x_n, t_n)]$}
    \KwOut{the robustness $\rob{\phi}(\signal, 1)$.}
    \BlankLine

        \tcp{Cells for atomic propositions can be computed concurrently.}
        \ShowLn\label{algoln:outerloop1}\ForPar{$i := n \KwTo 1$, $j := m \KwTo (m/2)$}{
            \Switch{$\phi_j$}{
                \lCase{$\top~/~\bot$}{$R[i, j] := +\infty~/~-\infty$}
                \lCase{$p(c)/~\neg p(c)$}{$R[i, j] := x_i - c/~ c - x_i$}
            }
        }
        \tcp{Fill bottom row. $d$ = Depth of $\formula_j$.}
        \ShowLn\label{algoln:outerloop2}\For{$d := (\log_2(m + 1) - 1) \KwTo 1$}{
            \ForPar{$j := (2^{d-1}) \KwTo (2^{d} - 1)$}{
                \Switch{$\phi_j$}{
                    \lCase{$\top~/~\bot$}{$R[n, j] := +\infty~/~ -\infty$}
                    \lCase{$p(c)~/~\neg p(c)$}{$R[n, j] := x_i - c~/~ c - x_i$}
                    \lCase{$\phi_{k_1} \lor \phi_{k_2}$}{$R[n, j] := \max(R[n, k_1], R[n, k_2])$}
                    \lCase{$\phi_{k_1} \land \phi_{k_2}$}{$R[n, j] := \min(R[n, k_1], R[n, k_2])$}
                    \uCase{$\phi_{k_1} \until{I} \phi_{k_2}$ \label{algoln:until bot row}}{
                        \lIf{$0 \in I$}{$R[n, j] := R[n, k_2]$}
                        \lElse{$R[n, j] := -\infty$}
                    }
                    \uCase{$\phi_{k_1} \release{I} \phi_{k_2}$ \label{algoln:release bot row}}{
                        \lIf{$0 \in I$}{$R[n, j] := R[n, k_2]$}
                        \lElse{$R[n, j] := +\infty$}
                    }
                }
            }
        }
        \tcc{Traverse remaining cells diagonally from bottom right to top left.}
        \ShowLn\label{algoln:diagonal}\For{$diag := (n + \log_2(m + 1) - 4) \KwTo 1$}{
            \ForPar{$d := (\log_2(m + 1) - 1) \KwTo 1$}{
                $i := diag - d$\;
                \ForPar{$j := (2^{d-1}) \KwTo (2^{d} - 1)$}{
                    \Switch{$\phi_j$}{
                        \lCase{$\top~/~\bot/p(c)/\neg p(c)/\phi_{k_1} \lor \phi_{k_2}/\phi_{k_1} \land \phi_{k_2}$}{same as above}
                        \uCase{$\phi_{k_1} \until{I} \phi_{k_2}$  \label{algoln:until}}{
                            $R[i, j] := -\infty$; $tmp_{min} := +\infty$\;
                            \ShowLn\label{algoln:singleinnerloop}\For{$h := i \KwTo n$}{
                                \If{$t_h \in t_i + I$}{
                                    $\tau = \min(tmp_{min}, R[h, k_2])$\;
                                    $R[i, j] := \max(\tau, R[i, j])$
                                }
                                $tmp_{min} := \min(R[h, k_1], tmp_{min})$
                            }
                        }
                        \uCase{$\phi_{k_1} \release{I} \phi_{k_2}$ \label{algoln:release}}{
                            $R[i, j] := +\infty$; $tmp_{min} := -\infty$\;
                            \For{$h := i \KwTo n$}{
                                \If{$t_h \in t_i + I$}{
                                    $\tau = \max(tmp_{max}, R[h, k_2])$\;
                                    $R[i, j] := \min(\tau, R[i, j])$
                                }
                                $tmp_{max} := \max(R[h, k_1], tmp_{max})$
                            }
                        }
                    }
                }
            }
        }
        \Return{$R[1, 1]$}
    \caption{Variant \variantI}
    \label{algo:variant1}
\end{algorithm}

Simply feeding the DP-TALIRO code to MP-SPDZ yields a very inefficient private calculation. 
So we first did a number of functionality-preserving modifications to DP-TALIRO, guided by an understanding of which MPC operations are efficient - yielding variant \variantI{} below. 
From there, we explored two orthogonal approaches, based on the observation that the max/min operations of the robust semantics are the biggest bottleneck:
first, we modified the robust semantics by replacing some of the max/min operations by additions, which are free in MPC.
Second, we restricted ourselves to the unbounded fragment of STL, which eliminates all $n$-ary max/min operations with $n>2$. 
Because these two approaches are orthogonal, they can be mixed-and-matched, yielding four variants in total: \variantI{} (with full STL and min/max semantics), \variantLIN{} (with linearized semantics), \variantUNB{} (for the unbounded STL fragment), and \variantLINUNB.
We start with the variant \variantI, which forms the basis of the other variants.

\subsection{\variantI: Algorithmic Optimizations for MPC}
\label{sec:variant v1}
\variantI{} is shown in \autoref{algo:variant1}. Its correctness as a robustness monitor is inherited from that of DP-TALIRO, as all modifications preserve functionality. 
\variantI{} focuses on parallelizations: as explained in Section \ref{sec:mpc background}, parallelization means reducing the number of operations that require sequential communication within the MPC.
Additionally, \variantI{} avoids conditionals that compute the same thing across multiple branches. 
That's because in MPC, all branches of a conditional are computed, because determining the correct branch by testing a secret value is forbidden.
This fact results in optimizations which would otherwise appear very unintuitive in traditional algorithm design.
\variantI{} introduces four significant modifications:
\begin{enumerate}[start=1,label={M\arabic*.}]
    \item Separation of the outer loop into three independent loops;
    \item Traversing the table diagonally along the depth of the formula;
    \item Removing the \bounds{} subroutine, building its operation directly into the algorithm; and
    \item Removing the Eventually/Always operators, which are inefficient in MPC. Instead, they are evaluated by re-writing them using the Until.
\end{enumerate}

\noindent
\emph{\textbf{M1: Separation of the outer loop.}\:}
We separate the outer loop of DP-TALIRO (\autoref{algoln:outer loop and diagonal} in \autoref{algo:dp-taliro}) into smaller loops at lines \ref{algoln:outerloop1}, \ref{algoln:outerloop2}, and \ref{algoln:diagonal} in \autoref{algo:variant1}.
This allows for increased parallelism. The first loop, \autoref{algoln:outerloop1}, computes all cells in the right half of the DP table concurrently, since they contain logical constants and atoms. This is the right half of \autoref{tab:diagonal table traversal example} page \pageref{tab:diagonal table traversal example}. 

The second loop in \autoref{algo:variant1}, \autoref{algoln:outerloop2}, computes the cell values for the last element of the signal, $(x_{\siglen}, t_{\siglen})$. We separate these because the temporal operators, evaluated at the last element of the signal, require fewer computations than they would elsewhere in the signal. In \autoref{tab:diagonal table traversal example}, this is the left half of the last row.
The third loop, \autoref{algoln:diagonal}, computes all cell values not otherwise addressed by the first and second loops, and is described in M2 below.

\begin{table}[t]
    \caption{A robustness table illustrating modifications M1 and M2 of \variantI{} (\autoref{algo:variant1}). As indicated, sub-formulas $\phi_2$ and $\phi_3$ are at depth 2 in the syntax tree, and $\phi_4-\phi_7$ are at depth 3. The first loop of $V1$ fills all the cells on the right half of the table concurrently; the second loop fills the cells in the last row concurrently; and the third loop iterates over the remaining cells in the order $\times, +, \bullet, \triangle, \star$, with all cells with a given symbol being filled concurrently.}
    \centering
    \setlength{\tabcolsep}{4pt} 
    \scriptsize
    \begin{tabular}{|c||c|c|c|c|c|c|c|c|c|c|c|}
        \hline
        & $\phi_1$ & $\phi_2$ & $\phi_3$ & $\phi_4$ & $\phi_5$ & $\phi_6$ & $\phi_7$ & $p$ & $q$ & $\dots\dots\dots$ & $\top$ \\ \hline
        & Depth 1 & \multicolumn{2}{c|}{Depth 2} & \multicolumn{4}{c|}{Depth 3} & \multicolumn{4}{c|}{Depth 4} \\
        \hline
        $t_1$ & \cellcolor{blue!15}$\star$ & \cellcolor{green!15}$\triangle$ & \cellcolor{green!15}$\triangle$ & \cellcolor{yellow!25}$\bullet$ & \cellcolor{yellow!25}$\bullet$ & \cellcolor{yellow!25}$\bullet$ & \cellcolor{yellow!25}$\bullet$ & \multicolumn{4}{c|}{\multirow{3}{*}{Atoms \& Constants:}} \\ \cline{1-8}
        $t_2$ & \cellcolor{green!15}$\triangle$ & \cellcolor{yellow!25}$\bullet$ & \cellcolor{yellow!25}$\bullet$ & \cellcolor{orange!25}$+$ & \cellcolor{orange!25}$+$ & \cellcolor{orange!25}$+$ & \cellcolor{orange!25}$+$ & \multicolumn{4}{c|}{\multirow{3}{*}{First Loop}} \\ \cline{1-8}
        $t_3$ & \cellcolor{yellow!25}$\bullet$ & \cellcolor{orange!25}$+$ & \cellcolor{orange!25}$+$ & \cellcolor{red!20}$\times$ & \cellcolor{red!20}$\times$ & \cellcolor{red!20}$\times$ & \cellcolor{red!20}$\times$ & \multicolumn{4}{c|}{} \\ \cline{1-8}
        $t_4$ & \multicolumn{7}{c|}{Last Row: Second Loop} & \multicolumn{4}{c|}{} \\ \hline
    \end{tabular}
    \label{tab:diagonal table traversal example}
\end{table}

\noindent
\emph{\textbf{M2: Traversing the table diagonally (\autoref{algoln:diagonal}).}\:}
The third loop increases parallel computation: rather than the row-by-row approach provided by \autoref{algo:dp-taliro}, \autoref{algo:variant1} traverses the left half of the table (which contains the robustnesses of non-atoms) in a diagonal fashion as shown in \autoref{tab:diagonal table traversal example}. 
With an input signal of length 4 and input formula $\bot \release{I_1} (\top \until{I_2} ((x \geq 0) \land \lnot(x \geq 10)))$ (depth 4, length 15), the loop has 5 iterations: in the first iteration it computes all cells marked by $\times$ concurrently, in the second it fills all cells marked by $+$ concurrently, etc. This significantly increases parallelism.

\noindent
\emph{\textbf{M3: Removing the \bounds{} subroutine.}\:}
In DP-TALIRO (\autoref{algo:dp-taliro}), the \bounds{} subroutine identifies the signal timestamps that fall within a given temporal interval $I$.
To do so it must compare timestamps with interval bounds, both of which are secret values. 
In MPC, no control flow based on testing secret values is allowed (as explained in Section \ref{sec:mpc background}).
So in MPC, all $n - i$ comparisons would have to be made, where $i$ is the current timestamp index. Furthermore, the loop after the \bounds{} subroutine iterates over the secret values returned by \bounds{} (see \autoref{algoln:secondinnerloop} in \autoref{algo:dp-taliro}), and thus must also perform $n - i$ iterations when done in MPC.
These two sets of $n - i$ iterations are highly inefficient in MPC.
So for the MPC variant \variantI{}, we combine these two sets into a single loop of $n - i$ iterations, e.g. \autoref{algoln:singleinnerloop} of \autoref{algo:variant1}. 
This loop checks if the iteration's timestamp is inside the bounds, and performs its calculations if so. This modification is applied to both Until and Release cases.


\noindent
\emph{\textbf{M4: Removing the Eventually and Always cases.}\:}
Because every branch of a switch statement is executed in MPC, a switch statement with fewer cases executes faster than one with more cases. 
Therefore, we remove the Eventually and Always cases (line \ref{algoln:unnecessary1} in DP-TALIRO). Any such operators in the input formula are first re-written using Until.

\noindent
\emph{\textbf{Complexity.}\:}
The complexity of this algorithm when executed in the clear (not in MPC) is $O(\siglen^2 \cdot \formulalen)$; the added complexity of \autoref{algo:variant1} compared to \autoref{algo:dp-taliro} is due to removal of data-dependent bounds on loops in the temporal operators.
This algorithm's sequential complexity, i.e. its amount of sequential operations, is $O(\siglen^2 + \siglen \cdot \log_2(\formulalen))$. This is from the $O(\siglen + \log_2(\formulalen))$ sequential cell traversal of the table (see \autoref{tab:diagonal table traversal example}), and the $O(n)$ operations per cell for the temporal operators.

Algorithm \variantI{} can be naturally transformed into an arithmetic circuit using MP-SPDZ, yielding the following. 

\begin{theorem}
    \variantI{} can be expressed as an arithmetic circuit of size $O(\siglen^2 \cdot \formulalen)$ over a suitably large field, and therefore there is an interactive MPC protocol that securely computes robustness with cost $O(\siglen^2 \cdot \formulalen)$.
\end{theorem}

\subsection{\variantVI: Leveraging Free Additions}
\label{sec:variant linear semantics}
Much of the computational cost of \variantI{} stems from the max/min operations which are used for evaluating every logical operator. 
These calculations are expensive due to their nonlinearity. 
We now propose a novel \emph{\linearsem{} semantics} of robustness, $\robapprox{\phi}(\signal, i)$, which replaces some (but not all) of the max/min calculations with linear operations. 
By using this new semantics in \variantI~instead of the max/min semantics, and slightly modifying the Until calculation as described below, we obtain variant \variantVI.

\begin{definition}\label{def:approx semantics}
Abbreviate $K(i) = \sigtimes^{-1}(\sigtimes(i) + I)$.
The \linearsem{} semantics of STL are as follows.
\begin{align}
        \robapprox{\top}(\signal, i) &\defeq +\infty ,\qquad  \robapprox{\bot}(\signal, i) \defeq 0\\
        \robapprox{x \geq c}(\signal, i) &\defeq \max(\sigvals(i) - c, 0)\\
        \robapprox{\lnot(x \geq c)}(\signal, i) &\defeq \max(c - \sigvals(i), 0) \\
        \robapprox{\phi_1 \lor \phi_2}(\signal, i) &\defeq \frac{\robapprox{\phi_1}(\signal, i) + \robapprox{\phi_2}(\signal, i)}{2} \label{eq:or semantics}\\
        \robapprox{\phi_1 \land \phi_2}(\signal, i) &\defeq \min(\robapprox{\phi_1}(\signal, i), \robapprox{\phi_2}(\signal, i)) \\
        \robapprox{\phi_1 \until{I} \phi_2}(\signal, i) &\defeq \frac{1}{|K(i)|} \sum_{j \in K(i)} \min \Bigl( \robapprox{\phi_2}(\signal, j) , \nonumber \\
        &\qquad \qquad\quad \min_{i \leq k < j} \robapprox{\phi_1}(\signal, k) \Bigr) \label{eq:until semantics}\\
        \robapprox{\phi_1 \release{I} \phi_2}(\signal, i) &\defeq \min_{j \in K(i)} \frac{1}{j - i + 1}\biggl( \robapprox{\phi_2}(\signal, j) ~+ \nonumber\\
        &\qquad \qquad\quad \sum_{i \leq k < j}\robapprox{\phi_1}(\signal, k) \biggr) . \label{eq:release semantics}
    \end{align}
\end{definition}

The soundness of the new semantics is captured in the following theorem, which is proved in Appendix~\ref{appendix:semanticsproof}.

\begin{theorem}\label{thm:linear sem props}
    Given a signal $\signal$, an index $i$, and formula $\phi$, the following hold
    \begin{enumerate}
        \item $\rob{\phi}(\signal, i) > 0 \implies 0 < \robapprox{\phi}(\signal, i) \leq \rob{\phi}(\signal, i)$
        \item $\rob{\phi}(\signal, i) \leq 0 \implies \robapprox{\phi}(\signal, i) = 0$.
    \end{enumerate}
\end{theorem}

\begin{algorithm}
    \uCase{$\phi_{k_1} \until{I} \phi_{k_2}$}{
        $tmp_{min} := [+\infty, \dots, +\infty]$
        \tcp{length $n - i$}
        $bounds := [0, \dots, 0]$
        \tcp{length $n - i$}
        \For{$h := i \KwTo n$}{
            $tmp_{min}[h - i + 2] := \min(R[h, k_1], tmp_{min}[h - i + 1])$\;
            $tmp_{min}[h - i + 1] := \min(R[h, k_2], tmp_{min}[h - i + 1])$\;
            \lIf{$t_h \in t_i + I$}{
                $bounds[h - i + 1] := 1$
            }
        }
        $R[i, j] := (bounds \cdot tmp_{min}) / (\sum bounds)$
        \tcp{dot product}
    }
    \caption{\variantVI: Until Case}
    \label{algo:variant6}
\end{algorithm}

Aside from using the new semantics, \variantLIN{} mostly follows \variantI{}. 
The Until case (\autoref{algo:variant1} \autoref{algoln:until}) is implemented slightly differently, taking advantage of slight efficiencies of the dot product operation compared to $\max$ when implemented in MPC (namely, reduced binary operations; see \autoref{sec:experiments} for results supporting these efficiencies). We show only this case in \autoref{algo:variant6} for space reasons; this case and other modifications match the semantics in \autoref{def:approx semantics}.

\noindent
\emph{\textbf{Complexity.}\:}
When evaluated in the clear (not in an MPC protocol), the complexity of \variantVI{} is the same as that of \variantI: $O(\siglen^2 \cdot \formulalen)$.
The runtime gains come when implemented in MPC.
Sequential operations are the same as \variantI{} as well: $O(\siglen^2 + \siglen \cdot \log_2(\formulalen))$.

The one-sided aspect of the semi-linear semantics, evident in \autoref{thm:linear sem props}, has an effect on duality -- for example, with formulas such as $\lnot p \land \lnot q$ and $p \lor q$, $\robapprox{\lnot p \land \lnot q}(s, i)$ is not the negation of $\robapprox{p \lor q}(s, i)$. We accept this unusual property of the semantics because we are testing until we find a bug; negative robustness values do not continue our testing loop, and so we only need positive values to inform testing. The experiments in the following section illustrate no penalty from using semi-linear semantics in a closed-loop falsification setting using simulated annealing.

\subsection{\variantIV: Restricting to Unbounded Temporal Operators}
\label{sec:variant unbounded}
Our next variant looks at one significant source of computational cost: evaluating Untils, which requires calling the BOUNDS routine, then doing $N$-ary max or min (where $N$ is a function of the interval size).
Therefore this variant, which is based on \variantI, restricts itself to formulas where all temporal operators are known to be of the form $[0,T]$, which is the largest possible interval.
We call it \variantIV{}.

In such a case, 
for the bottom row of the DP table (\autoref{algo:variant1} lines \ref{algoln:until bot row} and \ref{algoln:release bot row}), the Until evaluation is simplified to a single case $R[n, j] := R[n, k_2]$. 
For the other rows, the Until case (\autoref{algo:variant1} \autoref{algoln:until}) is replaced with $R[i, j] := \max(R[i, k_2], \min(R[i, k_1], R[i + 1, j]))$ and the Release case (\autoref{algo:variant1} \autoref{algoln:release}) is replaced with $R[i, j] := \min(R[i, k_2], \max(R[i, k_1], R[i + 1, j]))$. 

\noindent
\emph{\textbf{Complexity.}\:}
The complexity of this algorithm in the clear is $O(\siglen \cdot \formulalen)$, much less than that of \variantI. This is because for a single cell of the table we can compute a temporal operator in constant time, rather than the $O(\siglen)$ time required in \variantI.
The sequential complexity reduces similarly: $O(\siglen + \log_2(\formulalen))$. Table traversal does not become more or less sequential with \variantIV{} than with \variantI{}.

\subsection{\variantVII: Unbounded Fragment of STL with Semi-Linear Semantics}
\label{sec:variant lin-unb}
This last variant replaces the max/min semantics in \variantIV{} with the \linearsem{} semantics. 

\noindent
\emph{\textbf{Complexity.}\:}
Like \variantIV, the complexity of \variantVII{} is $O(\siglen \cdot \formulalen)$.
Sequential operations are likewise $O(\siglen + \log_2(\formulalen))$.

The four variants are summarized in \autoref{tab:variants}.

\begin{table}[t]
    \caption{The four variants and their complexities.}
    \centering
    \begin{tabular}{|c||c|c|c|}
        \hline
        Variant & Linear Robustness? & STL Fragment? & Complexity \\
        \hhline{|=#=|=|=|}
        \variantI & No & Full STL & $O(\siglen^2 \cdot \formulalen)$ \\
        \hline
        \variantVI & Yes & Full STL & $O(\siglen^2 \cdot \formulalen)$ \\
        \hline
        \variantIV & No & $I = [0, T]$ & $O(\siglen \cdot \formulalen)$ \\
        \hline
        \variantVII & Yes & $I = [0, T]$ & $O(\siglen \cdot \formulalen)$ \\
        \hline
    \end{tabular}    
    \label{tab:variants}
\end{table}

\subsection{Communication Complexity}
\label{sec:comm complexity}
The amount of data passed between parties (communication complexity) is dependent on the protocol used. The specifics for MASCOT are provided in \cite{cryptoeprint:2016/505}. 
For example, for pre-processing, Section 7.1 of \cite{cryptoeprint:2016/505} states that the communication cost of a single triple, which is a unit of data generated in pre-processing, is $n(n - 1)(\lambda k + k^2)$ bits where $n$ is the number of parties (2 in our case, Designer and Verifier), $\lambda$ is the security parameter (40 in our case, the default provided by MP-SPDZ), and $k$ is the number of bits of the modulo field (106 bits in our case). 
This suggests that the communication cost for a single pre-processing triple in our case is approximately 4 KB. Empirically, in our case the pre-processing data for signal length of 250 is approximately 4.2 GB. In the online phase, the data sent between parties is approximately 227 MB for that same signal length.
Other details on communication complexity can also be found in \cite{cryptoeprint:2016/505}.

\section{Experiments}\label{sec:experiments}

We aim to show that closed-loop testing is still practical when computing robustness privately. 
This involves two evaluations:
first, we must determine whether the runtimes of the various private variants are acceptable, and which variant is faster under which conditions.
Secondly, since variants \variantVI~and \variantVII~use the new \linearsem{} semantics, we must determine whether testing -- aka falsification -- still successfully finds bugs when guided by this new robustness notion.
(The other two variants use the classical semantics and so this question does not arise for them).
All experiments were conducted on a machine with a AMD Ryzen 7 PRO 7840U CPU and 16 GB of LPDDR5 RAM.

\subsection{Monitoring Results}
\label{sec:runtimes}
For implementing the robustness computation in MPC we use MP-SPDZ \cite{mp-spdz}, a commonly-used framework for MPC. From this framework we select the MASCOT \cite{cryptoeprint:2016/505} protocol, modified for semi-honest security. 
A run of this computation consists of three phases -- (1) pre-processing, which generates random data to be used in the MASCOT protocol; (2) setup, which performs initialization of the protocol; and (3) online, which runs the protocol with the provided inputs $\phi$ and $\signal$. 
Phases (1) and (2) do not require the inputs: (1) can be run far in advance, and (2) is run immediately before the provision of inputs. 
In closed-loop testing, (2) can be run while the model is being simulated to generate the trace $\signal$. 


In this set of experiments in MP-SPDZ we use three altitude signals from the FASER 05 dataset of \cite{faser05uav} with various lengths: 10, 50, 100, and 250. We use the following five formulas, where $T = 1023$.

\begin{enumerate}
    \item $\bot \>\release{[0,1023]} (\top \>\until{[0,10]} ((x \geq 250) \land \lnot (x \geq 330)))$
    \item $\top \>\until{[0,3]} (\bot \>\release{[5,20]} (\lnot (x \geq 330) \lor (x \geq 400)))$
    \item $(\bot \>\release{[2,4]} (\lnot (x \geq 250) \lor (x \geq 330))) \>\until{[0,1023]} (x \geq 400)$
    \item $(\top \>\until{[0,10]} \lnot (x \geq 250)) \lor (\top \>\until{[200,300]} (x \geq 330))$
    \item $\lnot (x \geq 250) \>\release{[2,4]} (\lnot (x \geq 400) \>\until{[30,32]} (\lnot (x \geq 330) \land (x \geq 250)))$
\end{enumerate}

We run all combinations of inputs, with the three signals of four different lengths each and five formulas, for each of our four algorithms, for a total of 240 combinations.
When producing results in this section, we average the online runtime across all runs for a given pairing of algorithm and signal length. \textit{Since execution is done within MPC, all inputs of the same size yield the same runtime and memory consumption}; the only differences in execution times or memory arise from input sizes and the algorithms being computed.



\begin{figure}[t]
        \centering
        \includegraphics[width=\linewidth]{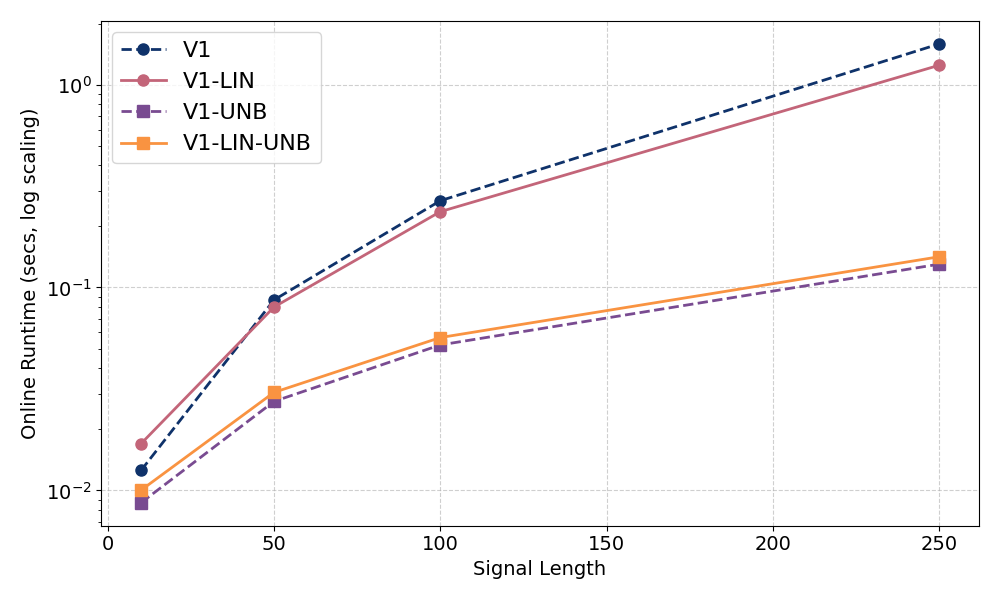}
        \caption{Online runtimes (seconds), log-linear scaling.}
        \label{fig:online runtime}
    \end{figure}
\begin{figure}[t]
        \centering
        \includegraphics[width=\linewidth]{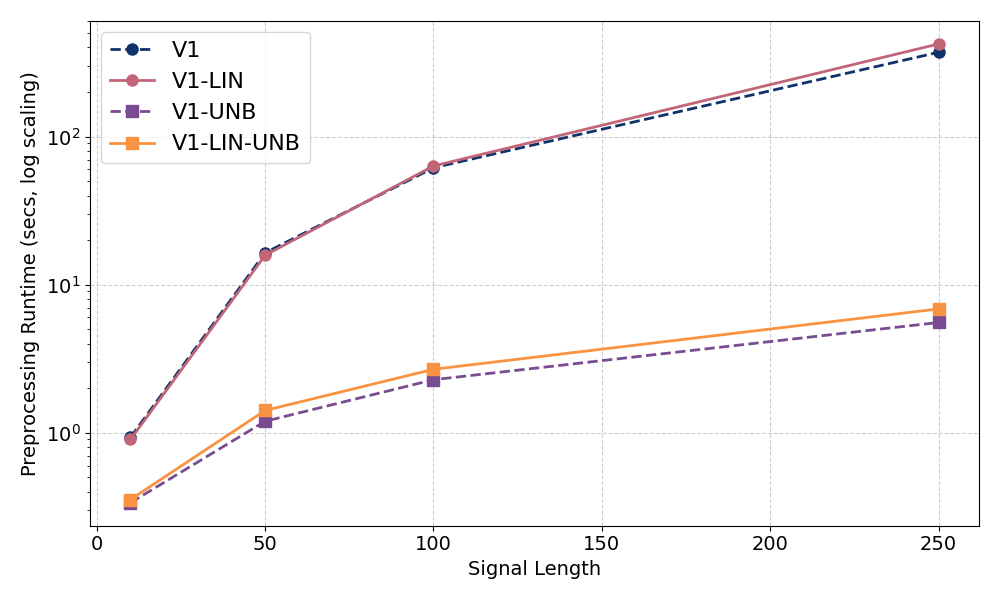}
        \caption{Preprocessing runtimes (seconds), log-linear scaling.}
        \label{fig:preprocessing runtime}
\end{figure}

\textbf{Online phase.}
In \autoref{fig:online runtime} we can see that the overall fastest variant is \variantUNB, with a runtime of 0.009 seconds at signal length 10 and 0.13 seconds at signal length 250.
For the bounded variants, the \linearsem{} semantics yield runtime savings; e.g. for signal length 250, 1.59 seconds vs 1.25 seconds for \variantI{} and \variantLIN, respectively.
But in the unbounded variants, the new semantics do not yield improvements (in fact it's a bit slower). This might be because the unbounded STL fragment only uses 2-place $\max/\min$ operations, and replacing some of these by a dot product is not noticeably faster, but rather requires extra time to convert to.

\begin{figure}[ht]
        \centering
        \includegraphics[width=\linewidth]{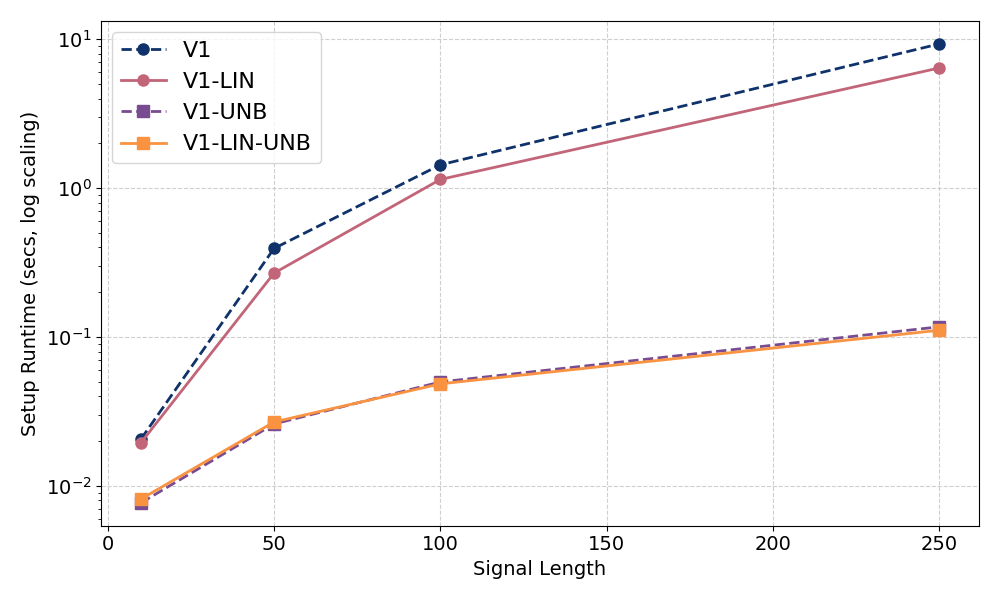}
        \caption{Setup runtime (seconds), log-linear scaling.}
        \label{fig:setup runtime}
\end{figure}
\begin{figure}[ht]
        \centering
        \includegraphics[width=\linewidth]{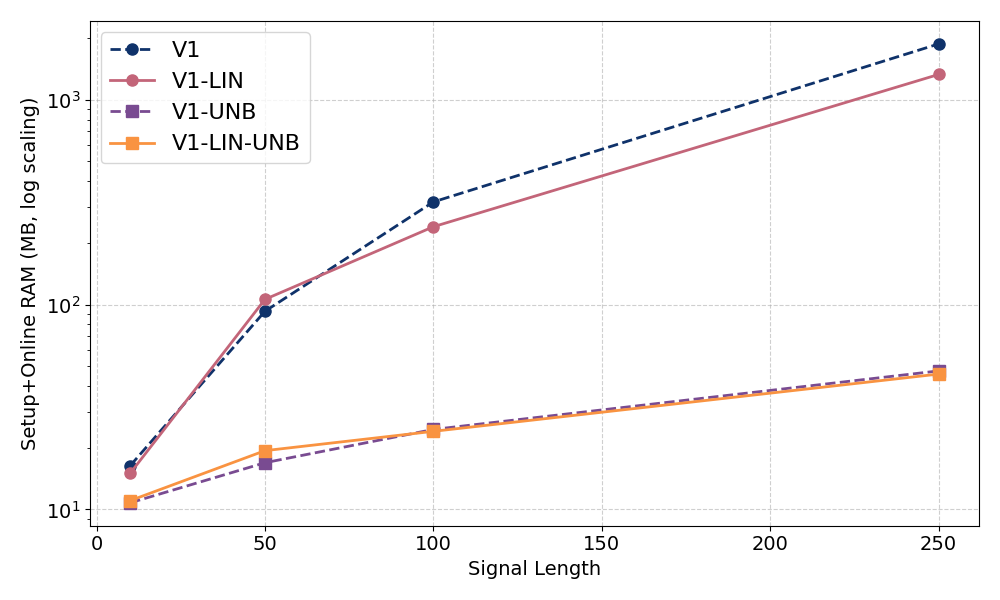}
        \caption{RAM of setup and online calculations (MB), log-linear scaling.}
        \label{fig:ram}
\end{figure}

\textbf{Pre-processing and setup phases.}
Figures \ref{fig:preprocessing runtime} and \ref{fig:setup runtime} show the runtimes for the pre-processing and setup phases, respectively. 
In the unbounded case, the \maxminsem{} semantics do better than the \linearsem{} semantics (at signal length 250, 5.55 seconds for \variantUNB{} vs 6.86 seconds for \variantLINUNB).
There is no significant difference between the semantics in the bounded case. 
For the setup phase, the only significant difference between semantics is in the bounded case -- \variantLIN{} beats \variantI{} with a time of 6.42 seconds versus 9.31 seconds, respectively for a signal length of 250. It is unclear what causes this increased performance; for setup in particular, performance gains are likely caused by implementation details within MP-SPDZ.
Finally, recall that pre-processing is done ahead of time independently of input formula and signal, and that setup can run in parallel with the system simulation. So as long as simulation is on the order of setup rutime, the overhead is negligible.

Across all runtime phases, we can see that \variantLIN{} is up to 30\% faster than \variantI{}, while \variantLINUNB{} is up to 20\% slower than \variantUNB. 
Moreover, the unbounded variants are considerably faster than the bounded variants. All runtimes grow polynomially as signal lengths increase (figures are log-linear scaling).

\textbf{RAM usage.}
The bounded variants use the most RAM: \variantLIN{} uses 1.3 GB for signal length 250, and \variantI{} uses 1.9 GB. 
The unbounded variants use considerably less -- about 46 MB for signal length 250.

These are trends similar to what we see with the runtimes. In addition to the unbounded variants outperforming the bounded variants, \variantLIN{} appears to outperform \variantI{} in RAM usage especially as signal lengths increase -- at signal length 250, \variantLIN{} uses approximately 45\% less RAM than \variantI{}, demonstrating that \variantLIN{} is useful for reducing both runtime and RAM usage. RAM scaling as a whole appears to match runtime scaling as signal lengths increase.


\subsection{Closed-Loop Testing}
\label{sec:closed-loop testing}
Robustness-guided testing, aka \textit{falsification}, searches for bugs by minimizing robustness over the space of system control signals: if it finds a control signal such that the system's output signal has negative robustness, then this is a specification violation.
We ran four system models: a bicycle, kite, quadcopter, and an F16 model from \cite{heidlauf2018verification}, We used three unique formulas per model, each with a depth 4 syntax tree. Formulas were chosen such that falsification rate was between 55\% and 90\% for \variantI{} after running a general simulated annealing loop with PSY-TaLiRo \cite{psytaliro} 1000 times. For the closed-loop testing we evaluated \variantI{} and \variantVI.

\begin{figure}[t]
    \centering
    \includegraphics[width=\linewidth]{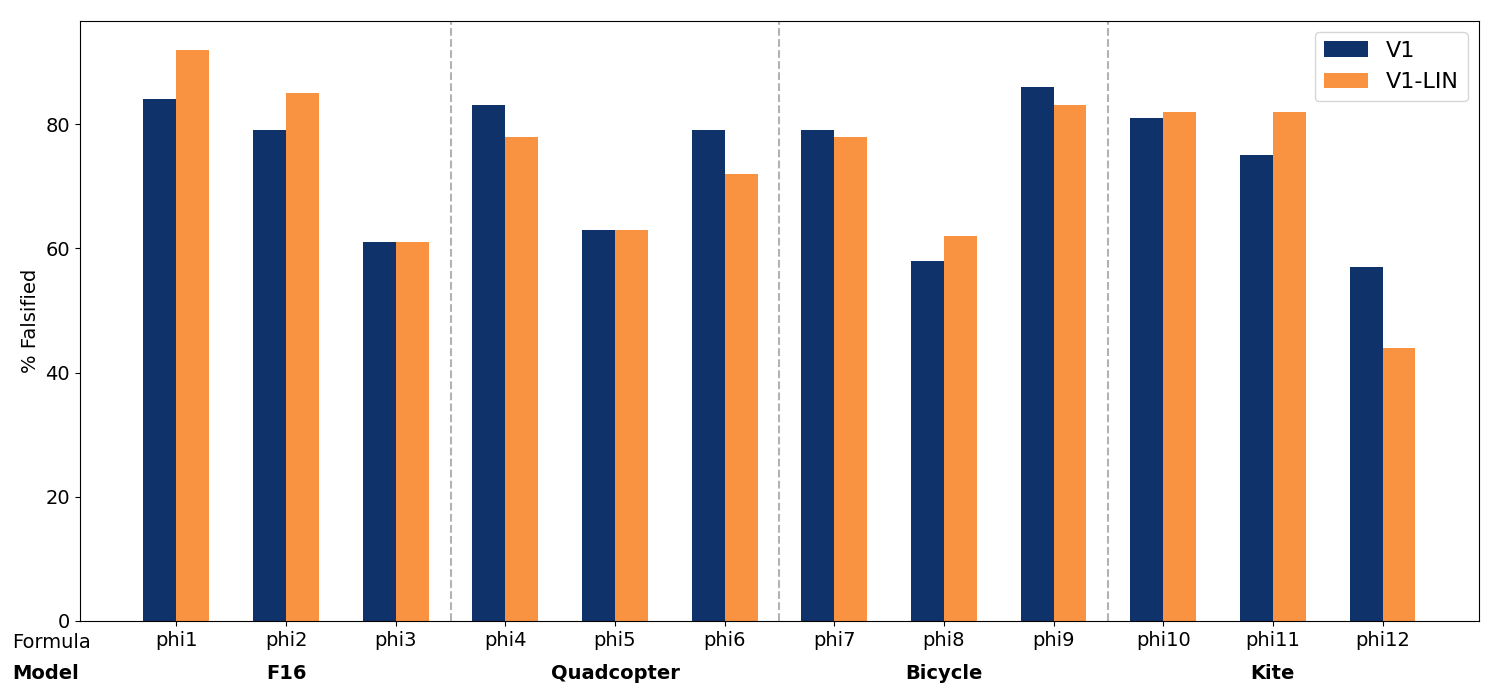}
    \caption{Falsification rate across models and formulas.}
    \label{fig:falsification}
\end{figure}

We minimized robustness using General Simulated Annealing within PSY-TaLiRo \cite{psytaliro}.
The minimization terminates when a negative robustness is found. 
We performed 100 trials of each system/formula/variant pairing, identifying 1) whether the trial produced a negative robustness, and 2) the iteration of Simulated Annealing at which there has been no change in the minimum robustness found so far in the past 300 iterations, suggesting a local minimum. 
We call this convergence speed. 
Falsification rates are shown in \autoref{fig:falsification}, and convergence speeds are shown in \autoref{fig:convergence}.

We conducted a two-sided Student's t-test for each model/formula pairing (tested data is shown in \autoref{fig:falsification}). The minimum p-value produced for a given model/formula pairing is 0.066 (the kite/phi12 pairing).
Thus, at the 0.05 significance threshold, there is no statistically meaningful evidence that \variantVI{} falsifies more or less often than \variantI.

\begin{figure}[t]
    \centering
    \includegraphics[width=\linewidth]{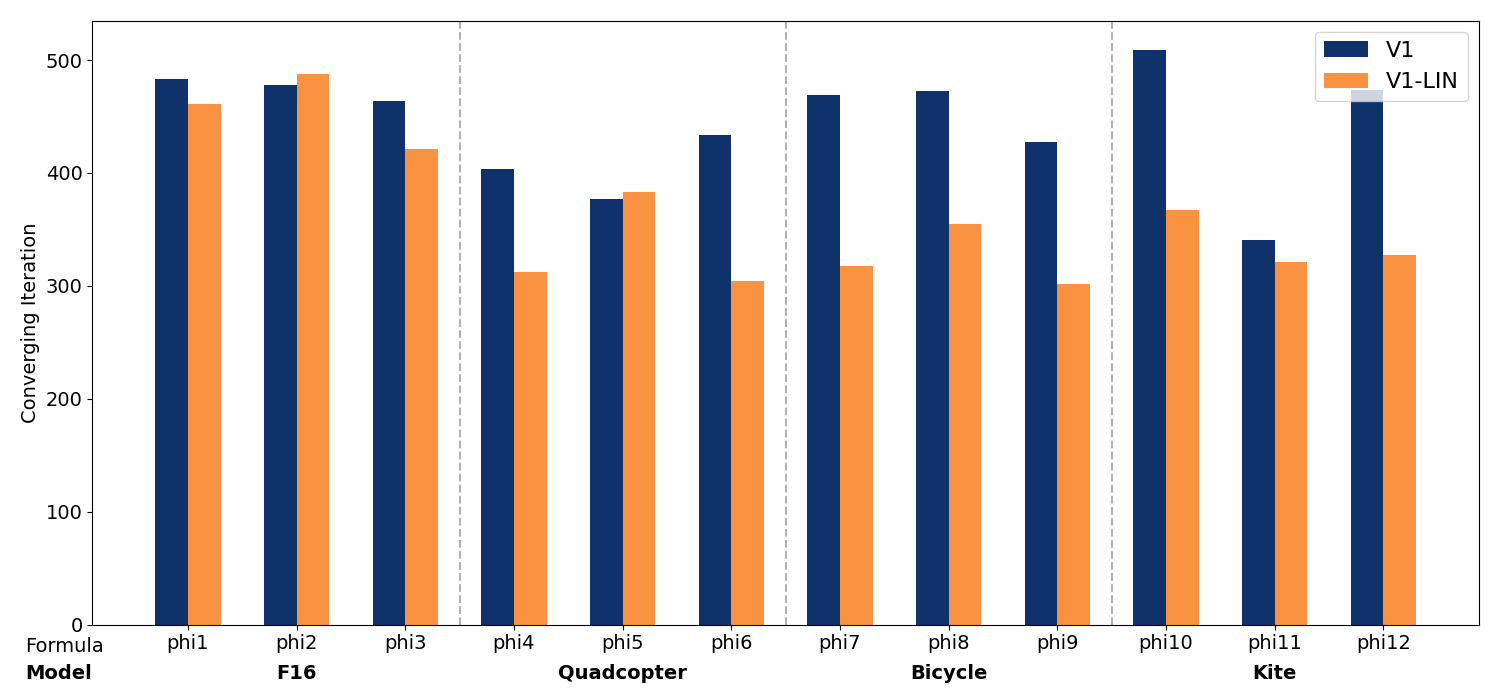}
    \caption{Convergence speed across models and formulas. The Y axis shows the iteration at which there has been no change in the minimum robustness found so far in the last 300 iterations.}
    \label{fig:convergence}
\end{figure}

For convergence testing we conducted a one-sided Student's t-test for each model/formula pairing (tested data is shown in \autoref{fig:convergence}), with a null hypothesis that \variantVI{} converges no more slowly than \variantI. The minimum p-value produced in these tests for a given model/formula pairing is 0.28 (the quadcopter/phi5 pairing).
There is no statistical evidence suggesting that \variantVI{} converges more slowly than \variantI. Rather, from our results, it seems as though \variantVI{} outperforms \variantI{} in converging to a minimum, though this wasn't an objective when designing the \linearsem{} semantics. 

Considering our results with falsification and convergence in the closed-loop setting and runtime/RAM capabilities in MPC, the \linearsem{} robustness \variantVI{} is a suitable replacement for \variantI{} in a closed-loop MPC setting.
Comparing the bounded and unbounded algorithms, the unbounded variants (\variantIV{} and \variantVII) significantly outperform the bounded ones (\variantI{} and \variantVI). This is expected, as the bounded variants address more complex formulas.

\section{Conclusion}\label{sec:conclusion}
We presented the first privacy-preserving robustness monitor for dense-time STL interpreted over timed state sequences. The monitor uses the MASCOT protocol as implemented in the MP-SPDZ MPC framework.
Working with the unbounded fragment of STL, and using a new, semi-linear, robust semantics were shown to yield significantly faster monitors.
Closed-loop testing appears practical for the \linearsem{} robustness measure, falsifying at least as quickly as the \maxminsem{} robustness.

\bibliographystyle{IEEEtran}
\bibliography{references}

\appendices
\section{Security Definitions and Guarantees}
\label{sec:security guarantees}

\newcommand{\ceq}{\overset{c}{=}}
\newcommand{\simulator}{\mathcal{S}}    
This section provides the formal definitions for protocol security that apply to our work. The material here is standard in the cryptography literature. 
We refer to \cite{evans2018pragmatic} and \cite{Cramer_Damgrd_Nielsen_2015} for textbook treatments of this material, and to \cite{canetti01UC} for a more in-depth treatment of the security model for MASCOT specifically.

We assume that there are two parties: \emph{Designer} and \emph{Tester}. 
An interactive MPC protocol, or simply a protocol, is a pair of instructions
$\pi = (\pi_{\mathsf{Des}}, \pi_{\mathsf{Mon}})$,
one for Designer and one for Tester, such that at each interaction round,
the Tester and the Designer send messages to each other as dictated by
the protocol. 
A protocol has a setup phase, in which the parties generate random data for later use in the protocol, among other things. We can think of the setup phase as round $0$.

We assume the two parties are honest-but-curious, aka semi-honest.
This means that they follow the steps of the protocol, but might store a transcript of the exchanged messages and their own random numbers, in order to perform extra calculations `on the side', to try and learn more than just $f(x,y)$. 

\paragraph{Designer’s private input.}
The Designer's input is a signal $\signal = (\sigvals, \sigtimes) \in [-C,C]^n \times [0,T]^n$ of length $n$, as per \autoref{def:signal}.
When doing calculations in the clear (without MPC) the signal values and timestamps are treated as type {\sc{double}}, say. 
For MPC, formally, the entries of $\sigvals$ and $\sigtimes$ are elements of a sufficiently large field, which allows to perform addition and multiplication with sufficient accuracy.

\paragraph{Tester’s private input.}
The Tester's input is an STL formula $\formula$ over a finite set of atomic propositions. 
This is represented as a list of all its sub-formulas.
Note that unlike in \cite{henzinger2025privacy} for instance, an STL specification has no automaton representation. 
Note also that the formula is an \textit{input} to the robustness functionality, whereas in \cite{henzinger2025privacy} it is fixed and compiled as part of the circuit.

\paragraph{Indistinguishability}
An \emph{ensemble} $X = \{X_\sigma\}_{\sigma \in \Ne}$ is a sequence of random variables $X_\sigma$, ranging over strings of bits of length polynomial in $\sigma$. For a given $\sigma$, the notation $x\sim X_\sigma$ will mean that $x$ is drawn according to the distribution of $X_\sigma$.
We will need the following notion to define security of a protocol. 
\begin{definition}
\label{def:comp indis}
    The ensembles $X$ and $Y$ are \emph{computationally indistinguishable} if for every probabilistic polynomial-time algorithm $D: \{0,1\}^* \rightarrow \{0,1\}$, every positive polynomial $p$, and all sufficiently large $\sigma$, it holds that 
    \[\left|Pr_{x \sim X_\sigma}[D(x)=1] - Pr_{y\sim Y_\sigma}[D(y) = 1]\right| < \frac{1}{p(\sigma)}\]
    If $X$ and $Y$ are computationally indistinguishable, we write $X \ceq Y$.
\end{definition}
Algorithm $D$ is known as the distinguisher.
Intuitively, $D$ tries to distinguish the two ensembles by computing a test ($D(x),D(y)$) that is sufficiently likely (probability more than $1/p(\sigma)$) to return different results on at least one random variable chosen from the two ensembles. 
The definition says that no such distinguisher exists for the two ensembles $X$ and $Y$: no matter what distinguisher is selected, there exist long enough random variables (via choice of $\sigma$) that it cannot distinguish with any likelihood. 

\paragraph{The ideal and real worlds.}
The definition of security refers to an ideal world, in which a trusted (or incorruptible) third party receives the inputs of Designer and Tester via secure channels, performs the computation $f(x,y)$, and reveals the result to the two parties. In this ideal setting, by definition, Designer and Tester can only learn whatever they can compute using their input and $f(x,y)$. This much is inevitable, which is why this is considered `ideal'.

In the real world, there is no trusted third party; rather, the two parties, Designer and Tester, execute the steps of the protocol $\pi$. What each of them gets to see is their own inputs, the messages they exchange with the other party, and any random numbers they draw. This constitutes their \emph{view}, $\mathsf{view}^{\pi}_{\mathsf{T}}(\formula, 1^\kappa)$ for Tester and $\mathsf{view}^{\pi}_{\mathsf{D}}(\signal, 1^\kappa)$ for Designer. 
Here $\kappa \in \Ne$ is the security parameter.
The protocol is considered secure (or private) if anything a party can compute based on their view in the real world is something they could compute anyway even in the ideal world. 
So we also define an ideal, or simulated, view for each party.
The ideal view for a party is whatever they can compute using a probabilistic polynomial-time algorithm (known as a \textit{simulator}) based on what they see in the ideal world, namely their own input and the final functionality output.
These ideal views are denoted $\simulator_{\mathsf{T}}(\formula, 1^\kappa, f(\formula,\signal))$ for Tester and  $\simulator_{\mathsf{D}}(\signal, 1^\kappa, f(\formula,\signal))$ for Designer.

Note that a view, whether real or ideal, is a random variable, so as we vary the inputs to the protocol, we obtain an ensemble.
To say that the real and ideal views are practically the same, we have to consider indistinguishability of the \textit{ensembles} of views. 
We thus arrive at the following definition.
\begin{definition}[2-party privacy w.r.t. semi-honest behavior and deterministic functionality]
    \label{def:security}
    A protocol $\pi$ privately computes a deterministic functionality $f$ iff there exists simulators $\simulator_{\mathsf{T}}$ and $\simulator_{\mathsf{D}}$ s.t. 
    \begin{eqnarray*}
        \{  \simulator_{\mathsf{T}}(\formula, 1^\kappa, f(\formula,\signal)) \}_{\formula, \signal, \kappa} &\ceq& \{\mathsf{view}^{\pi}_{\mathsf{T}}(\formula, 1^\kappa)\}_{\formula, \signal,\kappa}
        \\
        \{  \simulator_{\mathsf{D}}(\formula, 1^\kappa, f(\formula,\signal)) \}_{\formula,\signal, \kappa} &\ceq& \{\mathsf{view}^{\pi}_{\mathsf{D}}(\formula, 1^\kappa)\}_{\formula, \signal,\kappa}
    \end{eqnarray*}
\end{definition}
Note that the correctness of the protocol (the fact that it computes the correct value) is included in the definition of security: the trusted third party in the ideal world reveals the correct value $f(\formula,\signal)$, which is thus part of the ideal view. We pick a distinguisher $D$ (in \autoref{def:comp indis}) that simply returns the computed value from both ideal and real views. 
If the probability of error of the protocol cannot be made negligible (smaller than any inverse polynomial in $\kappa$) then this $D$ would be able to distinguish the two ensembles, which contradicts security.

\section{Proof of soundness of \linearsem{} robustness}\label{appendix:semanticsproof}
The following is a proof for \autoref{thm:linear sem props}.

\begin{proof}

    For clarity, let $v_j(i) \defeq \rob{\phi_j}(\signal, i)$ and $v_j'(i) \defeq \robapprox{\phi_j}(\signal, i)$ for $i \in \{0, 1, , \dots,n\}$ and $j \in \{1, 2\}$.

    \textbf{We start with the case $v \leq 0$}, and argue by induction on the structure of the formula. 
    Case $\phi = \top$ is excluded since $v\leq 0$ and case $\phi=\bot$ is true by definition.



    Case $\phi = p(c)$: We have that $\sigvals(i) - c \leq 0$. Thus $v' = \max(\sigvals(i) - c, 0) = 0$.

    Case $\phi = \lnot p(c)$: We have that $c - \sigvals(i) \leq 0$. Thus $v' = \max(c - \sigvals(i), 0) = 0$.

    Case $\phi = \phi_1 \lor \phi_2$: we have that $v_1(i) \leq 0$ and $v_2(i) \leq 0$. By the induction hypothesis, $v_1'(i) = v_2'(i) = 0$ and thus $v' = 0$.

    Case $\phi = \phi_1 \land \phi_2$: at least one of $v_1(i)$ and $v_2(i)$ is non-positive. Without loss of generality, say it is $v_1(i)$. Then by the induction hypothesis $v_1'(i) = 0$ and $v_2'(i) \geq 0$, and so $v' = 0$.

    Case $\phi = \phi_1 \until{I} \phi_2$: For all $j \in \sigtimes^{-1}(\sigtimes(i) + I)$, we have that $\min\lp v_2(j), \min_{i \leq k < j}\lp v_1(k)\rp\rp \leq 0$. Then by the same reasoning as the $\phi_1 \land \phi_2$ case, we have $\min\lp v_2'(j), \min_{i \leq k < j}\lp v_1'(k)\rp\rp = 0$. Summing these over all $j$ equates to 0, so $v' = 0$.

    Case $\phi = \phi_1 \release{I} \phi_2$: For at least one $j \in \sigtimes^{-1}(\sigtimes(i) + I)$, we have that $\max\lp v_2(j), \max_{i \leq k < j}\lp v_1(k)\rp\rp \leq 0$. Call the first instance of this occurrence $j^*$. Then $v_2'(j^*) = v_1'(i) = v_1'(i + 1) = \dots = v_1'(j^* - 1) = 0$. The summation of these is 0, so $v' = 0$.

    \textbf{Next we consider the case $v > 0$}. The base cases (logical constants and atoms) are immediate. 





    \underline{Case $\phi = \phi_1 \lor \phi_2$}: It holds that $\max(v_1(i),v_2(i))\geq 0$.
    Without loss of generality, suppose $v_1(i) \leq v_2(i)$, so a fortiori $0\leq v_1(i)$.
    By the induction hypothesis, $0 < v_1'(i) \leq v_1(i)$; if $v_2(i) \leq 0$ then $v_2'(i) = 0$ and if $v_2(i) > 0$ then $0 < v_2'(i) \leq v_2(i)$. Either way, $0 \leq v_2'(i) \leq v_1(i)$. 
    Since $v_1'(i) > 0$ and $v_2'(i) \geq 0$ the average of the two must be greater than 0, and since $v_1'(i) \leq v_1(i)$ and $v_2'(i) \leq v_1(i)$ then the average must also be less than $v_1(i)$. 
    Thus $0 < \frac{v_1'(i) + v_2'(i)}{2} \leq v_1(i)$. Therefore $0 < v'(i) \leq v(i)$.

    \underline{Case $\phi = \phi_1 \land \phi_2$}: Both $v_1(i) > 0$ and $v_2(i) > 0$, and so $0 < v_1'(i) \leq v_1(i)$ and $0 < v_2'(i) \leq v_2(i)$. The result follows from the induction hypothesis. 

    \underline{Case $\phi = \phi_1 \until{I} \phi_2$}: Let $a(j) \defeq \min(v_2(j), \displaystyle{\min_{i \leq k < j}}(v_1(k)))$ and let $a'(j) \defeq \min(v_2'(j), \displaystyle{\min_{i \leq k < j}}(v_1'(k)))$, for all $j \in \{1, \dots, n\}$. There is some $j^* \in \sigtimes^{-1}(\sigtimes(i) + I)$ where $\min(v_2(j^*), \displaystyle{\min_{i \leq k < j^*}}(v_1(k))) > 0$. 
    By similar logic to the previous case $\phi = \phi_1 \land \phi_2$, we have $0 < a'(j^*) \leq a(j^*)$.
    Then by similar reasoning to the case $\phi = \phi_1 \lor \phi_2$, we say the maximum of $a(j)$ is greater than or equal to the average of $a'(j)$, i.e., we have $\displaystyle{\max_{j \in \sigtimes^{-1}(\sigtimes(i) + I)}}(a(j)) \geq \frac{1}{|\sigtimes^{-1}(\sigtimes(i) + I)|} \sum_{j \in \sigtimes^{-1}(\sigtimes(i) + I)} a'(j)$. 
    Thus $v' \leq v$. 
    All other $a'(j)$ for $j \in \sigtimes^{-1}(\sigtimes(i) + I)$ are greater than or equal to 0, since every element in $\{v_2'(j), v_1'(i), v_1'(i+1),\dots,v_1'(j-1)\}$ is greater than or equal to 0 by the induction hypothesis. Therefore $v' > 0$, since every element of the average is at least 0 and $a'(j^*)$ is strictly positive. Summarizing, we have $0 < v' \leq v$.

    \underline{Case $\phi = \phi_1 \release{I} \phi_2$}: Let $a(j) \defeq \max(v_2(j), \displaystyle{\max_{i \leq k < j}}(v_1(k)))$ and let $a'(j) \defeq \frac{1}{j - i + 1}\lp v_2'(j) + \displaystyle{\sum_{i \leq k < j}} v_1'(k)\rp$, for all $j \in \{1, \dots, n\}$. 
    Since the maximum of values upper bounds their average, $a'(j) \leq a(j)$ for all $j$. 
    This holds for all $j$, so $\displaystyle{\min_{j \in \sigtimes^{-1}(\sigtimes(i) + I)}}(a(j)) \geq \displaystyle{\min_{j \in \sigtimes^{-1}(\sigtimes(i) + I)}}(a'(j))$ and so $v' \leq v$. 
    But we still need to show that $v' > 0$ -- to see this, consider that for all $j \in \sigtimes^{-1}(\sigtimes(i) + I)$, $a(j) > 0$. therefore at least one element in $\{v_2(j), v_1(i), v_1(i+1), \dots, v_1(j-1)\}$ is positive. Let $v_k(l)$ be the maximum of these elements; then $v_k(l) > 0$ and $v_k(l) = a(j)$. By the induction hypothesis we have that $0 < v_k'(l) \leq v_k(l)$. Thus there is at least one element in $\{v_2'(j), v_1'(i), v_1'(i+1), \dots, v_1'(j-1)\}$ that is positive, and so $a'(j)$ is also positive. The minimum of positive elements is positive, therefore $v' > 0$. Summarizing, we have $0 < v' \leq v$.
\end{proof}

\section{The \bounds{} subroutine}
\label{appendix:bounds}
Here we provide the pseudocode for the \bounds{} subroutine.
\begin{algorithm}
    \KwIn{a non-negative interval $I = [a, b]$, current index $i$, sampled times of the signal $\sigtimes := [t_1, \dots, t_n]$, $n$ the length of the signal, interval index lower bound $k_a$ from the previous iteration, interval index upper bound $k_b$ from the previous iteration}
    \KwOut{Return minimum and maximum indices of $\sigtimes$ values contained in the interval $t_i + I$}
    \BlankLine

    \Proc{BOUNDS($I, i, \sigtimes, n, k_a, k_b$)}{
        \lIf{$b = T$}{$k_b := n$}
        \Else{
            \leIf{$k_b = -\infty$}{$h := n$}{$h := k_b$}
            \For{$j := h \KwTo i$}{
                \If{$t_j < t_i + b$}{
                    $k_b' := j$\;
                    $break$\;
                }
            }
        }
        \leIf{$k_a = \infty$}{$h := n$}{$h := k_a$}
        \For{$j := h \KwTo i$}{
            \lIf{$t_j \geq t_i + a$}{$k_a' := j$}
            \lElse{$break$}
        }
        \Return{$k_a', k_b'$}
    }
    \caption{BOUNDS}
    \label{algo:bounds}
\end{algorithm}

\end{document}